\documentclass[runningheads]{llncs}

\usepackage{amsmath}
\usepackage{amsfonts}
\usepackage{amssymb}
\usepackage{graphicx}
\usepackage{todonotes}
\usepackage{booktabs}
\usepackage[switch]{lineno}
\usepackage{longtable}
\usepackage[colorlinks]{hyperref}
\usepackage{placeins}

\newcommand{\ignore}[1]{}

\newcommand{\cand}{\mathcal{C}}

\newcommand{\ballots}{\mathcal{B}}
\newcommand{\ballot}{b}
\newcommand{\election}{\mathcal{L}}

\DeclareMathOperator*{\basis}{basis}

\title{Risk-Limiting Audits for Condorcet Elections\thanks{This work was
supported by the University of Melbourne's Research Computing Services and the
Petascale Campus Initiative; and by the Australian Research Council (Discovery
Project DP220101012; OPTIMA ITTC IC200100009).}}

\author{
Michelle Blom     \inst{1}   \orcidID{0000-0002-0459-9917}  \and
Peter J. Stuckey  \inst{2}   \orcidID{0000-0003-2186-0459}  \and
Vanessa Teague    \inst{3,4} \orcidID{0000-0003-2648-2565}  \and
Damjan Vukcevic   \inst{5}   \orcidID{0000-0001-7780-9586}}

\authorrunning{Blom, Stuckey, Teague, Vukcevic}

\institute{
School of Computing and Information Systems, University of Melbourne,
Parkville, Australia.
\email{michelle.blom@unimelb.edu.au}
\and
Department of Data Science and AI, Monash University, Clayton, Australia
\and
Thinking Cybersecurity Pty.\ Ltd., Melbourne, Australia
\and
Australian National University, Canberra, Australia
\and
Department of Econometrics and Business Statistics, Monash University, Clayton,
Australia}

\begin{document}

\maketitle

\begin{abstract}
Elections where electors rank the candidates (or a subset of the candidates) in
order of preference allow the collection of more information about the
electors' intent.  The most widely used election of this type is Instant-Runoff
Voting (IRV), where candidates are eliminated one by one, until a single
candidate holds the majority of the remaining ballots.  Condorcet elections
treat the election as a set of simultaneous decisions about each pair of
candidates. The Condorcet winner is the candidate who beats all others in these
pairwise contests. There are various proposals to determine a winner if no
Condorcet winner exists. In this paper we show how we can efficiently audit
Condorcet elections for a number of variations.  We also compare the audit
efficiency (how many ballots we expect to sample) of IRV and Condorcet
elections.
\end{abstract}


\section{Introduction}

In ranked or preferential vote elections, each ballot comprises an ordered list
of (some or all of) the candidates.  The ballot is interpreted as a statement
that each candidate on the list is preferred by the voter to all the candidates
after it, and any that don't appear on the ballot.  Condorcet elections treat
each election as a series of two-way contests between each pair of candidates
$A$ and $B$, by saying $A$ beats $B$ if the number of ballots that prefer $A$
over $B$ is greater than those that prefer $B$ over $A$.  A Condorcet winner
exists if there is a single candidate that beats all other candidates.  But it
is quite possible to have ranked vote elections with no Condorcet winner. In
this case there are many alternate strategies/election systems that can be used
to choose a winner.

Risk-Limiting Audits (RLAs) test a reported election outcome by sampling ballot
papers and will correct a wrong outcome with high probability (by requiring a
full manual count of the ballots).  They will not change the outcome if it is
correct.  The \emph{risk limit}, often denoted $\alpha$, specifies that a wrong
outcome will be corrected with probability at least $1 - \alpha$.  RLAs for
Instant-Runoff Voting (IRV) can be conducted efficiently using RAIRE
\cite{blom2019raire}. RAIRE generates `assertions' which, if true, rule out all
outcomes in which the reported winner did not win. Assertions form the basis of
an RLA that can be conducted using the SHANGRLA framework \cite{shangrla}.

In this paper we show how to use the assertion-based methodology of Blom et
al.\ \cite{blom2021assertion} to form a set of assertions sufficient to conduct
an RLA for a variety of Condorcet elections.  Assertions with linear dependence
on transformations of the votes can easily be transformed to canonical assorter
form for SHANGRLA. We contrast the estimated difficulty of these audits, in
terms of sample sizes required, against auditing IRV using RAIRE.

For ranked vote elections that have a Condorcet winner, we first consider an
audit that checks that the reported winner is indeed the Condorcet winner
(\autoref{sec:Condorcet}). For an election with $n$ candidates, this requires
$n-1$ assertions comparing the reported winner to each reported loser. We then
consider Ranked Pairs, a Condorcet method that builds a preference relation
over candidates on the basis of the strength of pairwise defeats
(\autoref{sec:RankedPairs}). We find, that for elections with a Condorcet
winner, the expected sample sizes required to check that the reported winner is
the Condorcet winner, or to audit as a Ranked Pairs or IRV election, are
usually similar. In some instances, particularly those where the winner is
decided by who is eliminated in the second-last round of IRV, we see more
substantial differences in auditing difficulty. 

To demonstrate the practicality of our auditing methods, we use IRV datasets
from the Australian New South Wales (NSW) lower house elections in 2015 and
2019, in addition to a series of IRV elections held across the United States
between 2007 and 2010. All of these elections have a Condorcet winner.

Finding ranked vote datasets from real elections that did not have a Condorcet
winner was challenging. We were able to find some Single Transferable Vote
(STV) elections, and some datasets from Preflib\footnote{\url{www.preflib.org},
accessed 14 Mar 2023}, that met this criterion. For these instances, RAIRE
often struggled to terminate when finding an appropriate set of assertions to
audit. Auditing these elections as if they were Ranked Pairs elections was more
successful. We present these results in \autoref{sec:Results}.

We finally consider three additional Condorcet methods: Minimax
(\autoref{sec:Minimax}), Smith (\autoref{sec:Smith}), and Kemeny-Young
(\autoref{sec:KY}). We found that audits for these methods were generally not
practical. Minimax and Smith default to electing the Condorcet winner when one
exists, and in this case we can simply use the method outlined in
\autoref{sec:Condorcet}. On our instances without a Condorcet winner, we
generally did not find an audit for Minimax or Smith, with our proposed
methods, that was not a full manual count.


\section{Preliminaries}

In this paper, we consider \emph{ranked vote elections} that require voters to
cast a ballot in which they rank available candidates in order of preference,
and a single winner is determined. For example, in an election with candidates
$A$, $B$ and $C$, a voter indicates which of these candidates is their most
preferred, their second most preferred, and so on. Depending on the
jurisdiction, voters may be required to rank all candidates (e.g.,
[$A$,$B$,$C$]) or express a partial ranking (e.g., [$A$, $C$]). The latter
ballot is interpreted as expressing a preference for $A$ over $C$, and that
both $A$ and $C$ are preferred to all other candidates. The number of ballots
on which a candidate is ranked first is their first-preference tally.

We define a \emph{ranked vote election} as a pair $\election = (\cand,
\ballots)$ where $\cand$ is the set of candidates and $\ballots$ the multiset
of ballots cast. A ballot $\ballot$ is a sequence of candidates $\pi$, listed
in order of preference (most popular first), without any duplicates but also
without necessarily including all candidates. We use list notation to define
the ranking on a ballot (e.g., $\pi = [c_1,c_2,c_3,c_4]$). Given an election
$\election$, the election system determines which candidate $c \in \cand$ is
the winner.\footnote{Ties are also possible, but very rare for elections with
many ballots.}

\subsection{Instant-Runoff Voting (IRV)}\label{sec:Preferences}

IRV is a type of ranked vote election. To determine the winner in an IRV
election, each candidate is initially awarded all votes in which they are
ranked first (known as their \emph{first preferences}). The candidate with the
smallest tally is eliminated, and the ballots in their tally pile are
redistributed to the next most preferred candidate on the ballot who is still
standing. For example, a ballot with the ranking [$A$, $B$, $C$] is first
assigned to candidate $A$. If $A$ were to be eliminated, the ballot would be
given to candidate $B$, provided $B$ has not already been eliminated.
Subsequent rounds of elimination are performed in which the candidate with the
smallest tally is eliminated, and the ballots in their tally pile
redistributed, until we have one candidate left, or one of the remaining
candidates has the majority of votes. This candidate is declared the winner.

\begin{table}[t]
\centering
\caption{Distribution of ballot types for two example elections.}
\label{tab:Example}
\begin{tabular}{c@{~~~}c}
\begin{tabular}{|l|r|}
\hline
Ballot Signature & Votes \\
\hline
$A, B$    & 5,000 \\
$B, C$    & 2,500 \\
$C, A, B$ &   500 \\
$B, A$    &   300 \\
\hline
Total Votes & 8,300  \\
\hline
\end{tabular}
&
\begin{tabular}{|l|r|}
\hline
Ballot Signature & Votes \\
\hline
$A, C, B$ & 20,000 \\
$B, C, A$ & 19,000 \\
$C$       &  5,000 \\
\hline
Total Votes & 44,000   \\
\hline
\end{tabular} \\
(a) Election 1  &
(b) Election 2
\end{tabular}
\end{table}

\begin{table}[t]
\centering
\caption{Election 3 example: (a) distribution of ballot types; (b) calculation
of comparative tallies $T(i \succ j)$ for row $i$ and column $j$; (c)
calculation of scores $s(i,j)$.}
\label{tab:Example2}
\begin{tabular}{c@{~~~}c@{~~~}c}
\begin{tabular}{|l|r|}
\hline
Ballot Signature & Votes \\
\hline
$A, B, D, C$ & 7,000 \\
$A, C, B, D$ & 2,000 \\
$B, C, D, A$ & 4,000 \\
$B, D, A, C$ & 6,000 \\
$C, A, B, D$ & 2,000 \\
$C, D, A, B$ & 7,000 \\
$D, C, A, B$ & 1,000 \\
\hline
Total Votes & 27,000 \\
\hline
\end{tabular} \
&
\begin{tabular}{|l|rrrr|}
\hline
$T(i \succ j)$ & $A$ & $B$ & $C$ & $D$ \\
\hline
$A$ &     & 19k & 15k & 11k \\
$B$ & 10k &     & 17k & 21k \\
$C$ & 14k & 12k &     & 15k \\
$D$ & 18k &  8k & 14k &     \\
\hline
\end{tabular}
&
\begin{tabular}{|l|rrrr|}
\hline
$s(i,j)$ & $A$ & $B$ & $C$ & $D$ \\
\hline
$A$ &       &     9k &    1k & $-$7k \\
$B$ & $-$9k &        &    5k &   13k \\
$C$ & $-$1k &  $-$5k &       &    1k \\
$D$ &    7k & $-$13k & $-$1k &       \\
\hline
\end{tabular} \\
(a) & (b) & (c) 
\end{tabular}

\end{table}

Consider the distribution of votes present in two example elections shown in
\autoref{tab:Example}. In Election 1, candidates $A$, $B$, and $C$ have first
preference tallies of 5000, 2800, and 500, respectively. When viewed as an IRV
election, candidate $A$ has the majority of votes on first preferences and is
declared the winner.  In Election 2, $A$, $B$, and $C$ start with 20000, 19000,
and 5000 votes, respectively. Candidate $C$ is eliminated, with all of their
ballots \textit{exhausted}. Candidate $A$ has more votes than $B$ and is
declared the winner.  For the distribution of ballots for Election 3 shown in
\autoref{tab:Example2}(a), when viewed as an IRV election, candidate $D$ is
eliminated first, then $A$, leaving $B$ to win with a majority of votes.



\section{Risk-Limiting Audits for Condorcet Winners}
\label{sec:Condorcet}

In \emph{Condorcet elections} we consider the ballots as applying to
``separate'' decisions about each  pair of candidates $i,j \in \cand$.  A
ballot $b$ prefers $i$ over $j$ if $i$ appears before $j$ in $b$, or $i$
appears on $b$ and $j$ does not.  We define $T(i \succ j)$ to be the tally
(i.e., number) of ballots $b \in \ballots$ that prefer $i$ over $j$. Similarly,
we define $T(j \succ i)$ to be the tally of ballots $b \in \ballots$ that
prefer $j$ over $i$.  Note that $T(i \succ j) + T(j \succ i) \leqslant
|\ballots|$ as some ballots may mention neither $i$ nor $j$.

For the competition between a pair of candidates $i,j \in \cand$, we define a
\emph{score},
\begin{equation}
    s(i, j) = T(i \succ j) - T(j \succ i). \label{eqn:Strength}
\end{equation}

In an election that satisfies the Condorcet winner criterion, the winner is the
candidate $w \in \mathcal{C}$ for which $s(w,c) > 0$ for all $c \in \mathcal{C}
\setminus \{w\}$.  In Election 1 of \autoref{tab:Example}, candidate $A$ is the
Condorcet winner, because $T(A \succ B)$ is 5500, $T(B \succ A)$ is 2800, $T(A
\succ C)$ is 5300,  and $T(C \succ A)$ is 3000; giving $s(A,B) = 2700$ and
$s(A,C) = 2300$.  In Election 2 candidate $C$ is the Condorcet winner since
$T(A \succ C)$ is 20000, and $T(C \succ A)$ is 24000; $T(B \succ C)$ is 19000,
and $T(C \succ B)$ is 25000; giving $s(C,A) = 4000$ and $s(C,B)$ = 6000. Note
how a Condorcet winner need not be the IRV winner (which is $A$).  For Election
3, shown in \autoref{tab:Example2}(a), the tallies $T(i \succ j)$ for row $i$
and column $j$ are shown in \autoref{tab:Example2}(b), and the scores $s(i,j)$
for row $i$ and column $j$ are shown in \autoref{tab:Example2}(c).  There is no
Condorcet winner since one of the scores in each row is negative.

We can audit the correctness of the winner in these elections by checking the
assertions $s(w,c) > 0$ for all $c \in \mathcal{C} \setminus \{w\}$.  This is
similar to auditing first-past-the-post elections where we compare the tallies
of two candidates, but rather than comparing tallies of ballots for each
candidate we compare tallies of ballots that favor one candidate over the
other.  The tally of $w$, when compared with candidate $c$, is given by $t_w =
T(w \succ c)$, and that of $c$ is $t_c = T(c \succ w)$.  Our assertion states
that $t_w > t_c$ which can be transformed into an assorter as described by Blom
et al.\ \cite{blom2021assertion}, and audited with SHANGRLA \cite{shangrla}.
These assertions can be used to audit any Condorcet method where a Condorcet
winner exists. We simply have to check that the reported winner is the
Condorcet winner. We cannot do this for the election of \autoref{tab:Example2},
as it does not have a Condorcet winner.  If the reported results indicate a
Condorcet winner, but this is not actually the case, then at least one of
assertions will be false, so with probability at least $1-\alpha$ the audit
will not certify. Hence this is a valid RLA.


\section{Risk-Limiting Audits for Ranked Pairs Elections}
\label{sec:RankedPairs}

Ranked Pairs \cite{tideman1987independence} is a type of Condorcet election
that determines a winner even if no Condorcet winner exists.  Ranked Pairs
elections build a preference relation among candidates. First, we compute a
score for each pair of candidates, as per \autoref{eqn:Strength}.  Pairs with a
positive score are called \textit{positive majorities}.   We consider each of
these positive majority pairs in turn, from highest to lowest in score, and
build a directed acyclic graph $\mathcal{G}$ containing a node for each
candidate, and edges representing preference relationships. When we consider a
pair, ($i$, $j$), we add an edge from $i$ to $j$ in $\mathcal{G}$, if there
does not already exist a path from $j$ to $i$.  If such a path  exists, this
means that the preference $i \succ j$ is inconsistent with the stronger
preference relationships already added to $\mathcal{G}$. We therefore ignore
pair ($i$, $j$) and move on to the next pair. If ever there is a candidate $w$
s.t.\ there is a path in $\mathcal{G}$ from $w$ to all others, we declare $w$
the winner.

Consider Elections 1 and 2 in \autoref{tab:Example}. As Ranked Pairs elections,
we assign the scores shown in \autoref{tab:ExampleStrengths} to each pair.  For
Election 1, the sorted positive majorities are  ($B$,$C$), ($A$,$B$), and ($A$,
$C$).  We first add $B \succ C$, and then $A \succ B$, to $\mathcal{G}$. At
this point, we have established that $A$ is preferred to $B$, and by
transitivity, that $A$ is preferred to $C$. We can stop at this point, as we
have established enough preference relationships to declare $A$ as the winner,
and cannot add a later preference, or generate a new transitive inference, that
will be inconsistent with these relationships.  In Election 2, the sorted list
of positive majorities is ($C$,$B$), ($C$,$A$), and ($A$, $B$). After the first
two preference relationships are added to $\mathcal{G}$, we have established
that $C$ is the winner. The Ranked Pairs winner always coincides with the
Condorcet winner if one exists.

\begin{table}[t]
\centering
\caption{Pairwise scores for pairs in Elections 1 and 2 of
\autoref{tab:Example}.}
\label{tab:ExampleStrengths}
\begin{tabular}{|l|r|r|}
\hline
 & Election 1 & Election 2 \\
\hline
$s(A,B)$ &    2700 &    1000 \\ 
$s(B,A)$ & $-$2700 & $-$1000 \\ 
$s(A,C)$ &    2300 & $-$4000 \\ 
$s(C,A)$ & $-$2300 &    4000 \\ 
$s(B,C)$ &    7300 & $-$6000 \\ 
$s(C,B)$ & $-$7300 &    6000 \\ 
\hline
\end{tabular}
\end{table}

Consider Election 3 shown in \autoref{tab:Example2}.  The ranked positive
majorities are $(B,D), (A,B), (D,A), (B,C), (C,D),$ and $(A,C)$.  We add $B
\succ D$ and then $A \succ B$ to $\mathcal{G}$. We ignore $(D,A)$ since we have
already inferred $A \succ D$ by transitivity. We then add $B \succ C$ to
$\mathcal{G}$. We now have enough information to declare $A$ the winner, as
they are preferred to all other candidates through transitivity. Note that when
treated as an IRV election, $B$ is the winner. 

To audit a Ranked Pairs election, we must check that all preference statements
between pairs of candidates that were \textit{ultimately used} to establish
that the reported winner $w$ won do in fact hold.  Denote the set of preference
relationships we commit to (i.e., that we add to $\mathcal{G}$) by the Ranked
Pairs tabulation process up to the point at which the winning candidate is
established as $\mathcal{M}$.  In Election 1 of \autoref{tab:Example}, the
winner is established after the first two commits. Thus, $\mathcal{M}$ $=$ \{$B
\succ C$, $A \succ B$\}. For Election 2 of \autoref{tab:Example}, $\mathcal{M}
= \{C \succ B, C \succ A \}$. For Election 3 in \autoref{tab:Example2},
$\mathcal{M} = \{ B \succ D, A \succ B, B \succ C\}$.

Let $\mathcal{T}$ denote the set of \textit{transitively inferred} preferences
$i \succ j$ that were made in the tabulation process up to the point at which
the winner is established. Each such inference is associated with a path in
$\mathcal{G}$ from $i$ to $j$. We denote the set of preferences that were used
to form the transitive inference $i \succ j$ as $\basis(i \succ j)$. This
consists of all preference relationships along a path from $i$ to $j$ in
$\mathcal{G}$.  In Election~1 of \autoref{tab:Example}, $\mathcal{T} = \{A
\succ C\}$ and $\basis(A \succ C) = \{A \succ B, B \succ C\}$. For Election 3
in \autoref{tab:Example2}, $\mathcal{T} = \{A \succ D, A \succ C\}$. For these
transitive inferences, $\basis(A \succ D) = \{A \succ B, B \succ D\}$, and
$\basis(A \succ C) = \{A \succ B, B \succ C\}$.

We now define the assertions $\mathcal{A}$ required to audit a Ranked Pairs
election.  First, for each $w \succ j \in \mathcal{M}$, where $w$ is the
reported winner,  we must check that  $s(w, j) > 0$.  This corresponds to
checking that $(w,j)$ is a positive majority, where $T(w \succ j) - T(j \succ
w) > 0$.  This can be achieved as outlined in \autoref{sec:Condorcet}.

Next, we must check that any transitive inference $w \succ c$, from
$\mathcal{T}$, that we used to declare $w$ the winner could not have been
contradicted by a pair $(c, w)$ that, in the true outcome, was actually
stronger than one or more of the preferences used to infer $w \succ c$ in the
reported outcome. In Example 1 of \autoref{tab:Example}, this could occur if
the pair $(C,A)$ actually had a strength of 8000, for example, in place of
$-2300$. If this were the case, $C \succ A$ should have been the first
preference committed, ultimately leading to $B$ being declared the winner in
place of $A$. 

Note that in Ranked Pairs elections where we have a Condorcet winner, we could
simply verify that the reported winner was the Condorcet winner, irrespective
of whether transitive inferences were used in the tabulation process. This
would remove the need for an additional type of assertion, which we present in
the next section. For most election instances we consider in our Results
(\autoref{sec:Results}), our Ranked Pairs auditing method reduces to checking
the set of assertions for verifying that the reported winner is the Condorcet
winner (see \autoref{sec:Condorcet}). This is because transitive inferences
were not used to establish the winner in these cases.  However, checking these
transitive inferences, when they are used, through the assertions developed in
the next section, can sometimes be more efficient.

\subsection{Assertions and Assorters for Transitive Inferences}

For each transitive inference of the form $w \succ c \in \mathcal{T}$, we must
check that:
\begin{equation*}
    s(i,j) > s(c,w),  \quad \forall i \succ j \in \basis(w \succ c).
\end{equation*}
We can translate this expression into the following:
\begin{eqnarray}
    T(i \succ j) - T(j \succ i) & > & T(c \succ w) - T(w \succ c) \notag \\
    T(i \succ j) + T(w \succ c)   -   T(c \succ w) - T(j \succ i)
                                & > & 0 \label{eqn:Tally1}
\end{eqnarray}
We use the approach of Blom et al.\ \cite{blom2021assertion} to construct an
assorter for an assertion of the form shown in \autoref{eqn:Tally1}. We first
form a proto-assorter
\begin{equation*}
    g(b) = b_1 + b_2 - b_3 - b_4,
\end{equation*}
where $b$ is a ballot and each $b_i$ is the number of votes the ballot $b$
contributes to the category $t_i$, where $t_1 = T(i \succ j)$, $t_2 = T(w \succ
c)$, $t_3 = T(c \succ w)$, and $t_4 = T(j \succ i)$. We then form an assorter,
$h(b)$, for our assertion in \autoref{eqn:Tally1} using Equation 2 of Blom et
al.\ \cite{blom2021assertion}. This equation states that $h(b) = \frac{g(b) -
a}{-2a}$ where $a$ denotes the minimum value of $g(b)$ for any $b$. In our
case, $a = -2$, giving
\[
    h(b) = \frac{g(b) + 2}{4}.
\]
The assorter $h$ calculates the mean score $\bar{h}$ over the ballots $b$
examined in an audit.  By construction $\bar{h} > 1/2$ if and only if the
assertion $s(i,j) > s(c,w)$ holds, to make it fit into the SHANGRLA testing
framework \cite{shangrla}.

\subsection{Correctness of Audit Assertions}

The assertions in the Ranked Pairs audit are then:
\begin{align}
\mathcal{A} =
 &\{ s(w,c) > 0      : w \succ c \in \mathcal{M} \} \, \cup \\ \nonumber
 &\{ s(i,j) > s(c,w) : i \succ j \in \basis(w \succ c), w \succ c\in \mathcal{T}\}.
\end{align}
We now show that if these assertions are all verified to risk limit $\alpha$
then the declared winner is correct with risk limit $\alpha$.

\begin{theorem}
If the declared winner $w$ is not the correct winner of a Ranked Pairs
election, then the probability that an audit verifies all the assertions in
$\mathcal{A}$ is at most $\alpha$, where $\alpha$ is the risk limit of the
audit of each individual assertion.
\end{theorem}
\begin{proof}
Let $\mathcal{M}$ be the set of ranked pairs committed to $\mathcal{G}$ in the
reported election, $\mathcal{T}$ the transitively inferred preferences from
$\mathcal{M}$.  Let $\mathcal{M}'$ be the set of ranked pairs committed to
$\mathcal{G}$ for the actual election results, and $\mathcal{T}'$ the
transitively inferred preferences from $\mathcal{M}'$.  Assume that $w' \neq w$
is the actual winner.

Suppose that $w \succ w' \in \mathcal{T}$, so we audit $s(i,j) > s(w',w)$ for
all $i \succ j \in basis(w \succ w')$.  If all of these facts were correct then
the ranked pair voting algorithm (on the true counts) would find $w \succ w'$
by transitive closure. Contradiction. So at least one of them must not hold,
and the audit will accept it with probability at most $\alpha$. 

Otherwise $w \succ w' \in \mathcal{M}$. If $s(w',w) > 0$ then this contradicts
the audited assertion, which will be accepted with probability at most
$\alpha$.  So suppose $s(w',w) < 0$. Since $w'$ beats $w$ there is a $w''$ such
that $w' \succ w'' \in \mathcal{T}'$ and $w'' \succ w' \in \mathcal{M}'$ and
$s(w'',w) > s(w,w') > 0$.  Now in the reported election either $w \succ w'' \in
\mathcal{M}$ in which case the assertion $s(w,w'') > 0$ will be accepted with
probability at most $\alpha$, or $w \succ w'' \in \mathcal{T}$ and we use the
argument of the previous paragraph to show that the audit will accept with
probability at most $\alpha$.
\qed
\end{proof}

Consider Election 3 shown in \autoref{tab:Example2}.
In the Ranked Pairs election we established $A$ as the winner using
$\mathcal{M} = \{B \succ D, A \succ B, B \succ C\}$ and
$\mathcal{T} = \{A \succ D, A \succ C\}$, where
$\basis(A \succ D) = \{ A \succ B, B \succ D \}$ and
$\basis(A \succ C) = \{ A \succ B, B \succ C \}$.
So the assertions we need to verify are $s(A,B) > 0$;
$s(A,B) > s(D,A)$ and $s(B,D) > s(D,A)$; and
$s(A,B) > s(C,A)$ and $s(B,C) > s(C,A)$.

Note that Tideman describes a particular approach for resolving ties between
sorted majorities \cite{tideman1987independence}. If the choice of which
majority to process first, among those that tie, changes the ultimate winner
then a full manual hand count will be required. The manner in which such ties
are resolved can have an impact on the overall set of assertions formed. For
example, consider a case where we have three positive majorities ($A$, $B$),
($A$, $C$), and ($B$, $C$) in a three-candidate election. The first is the
strongest, while the latter two tie. Of the latter two, if we choose to process
($A$, $C$) first, then our audit will form two assertions: $s(A, B) > 0$ and
$s(A, C) > 0$. If we choose to process ($B$, $C$) first then we will form the
assertions: $s(A, B) > 0$, $s(A,B) > s(C,A)$ and $s(B,C) > s(C,A)$.


\section{RLAs for Minimax Elections}
\label{sec:Minimax}

In a Minimax election, a pairwise score is computed for each pair of
candidates, $c$ and $c'$, denoted $ms(c, c')$. There are variations on how this
score can be
defined\footnote{\url{https://en.wikipedia.org/wiki/Minimax_Condorcet_method},
accessed 14 Mar 2023}. One method, denoted \textit{margins}, is defined by:
\[
    ms(c, c') = T(c \succ c') - T(c' \succ c).
\]
This score computation method is equivalent to that used in Ranked Pairs. We
use this approach when forming assertions to audit Minimax elections.
Variations could be used instead, however their equations must be linear to be
used within the assertion-assorter framework of Blom et al.\
\cite{blom2021assertion}.
%

If there is a Condorcet winner, then this candidate wins. Suppose it is
candidate $w$, then we have $ms(w, c) > 0$ for all $c \in \mathcal{C} \setminus
\{w\}$ when either the margins or winning votes scoring method is used.
Otherwise, the winning candidate is the one with the smallest loss in pairwise
contests with each other candidate.

Consider the example elections in \autoref{tab:Example}.
\autoref{tab:ExampleStrengths} records the pairwise scores for each pair of
candidates. In both of our example elections, candidate $A$ is the Condorcet
winner. In Election 1, $ms(A, B) = 2700$ and $ms(A, C) = 2300$. In Election 2,
$ms(C, B) = 6000$ and $ms(C, A) = 4000$.

In the case where we have a Condorcet winner $w$ under Minimax, we simply need
to audit that $ms(w, c) > 0$ for all $c \in \mathcal{C} \setminus \{w\}$, as
described in \autoref{sec:Condorcet}.

In the case where we do not have a Condorcet winner, we compute each candidate
$c$'s largest margin of loss, $LL(c)$. The candidate $c$ with the smallest
$LL(c)$ is the winner. Consider a case with the pairwise scores shown below
(left). In this example, we have the largest losses shown on the right. In this
case, candidate $B$ has the smallest largest loss and is the winner.

\begin{longtable}{p{110pt}p{110pt}}
$ms(A, B) = 2000$ & $LL(A) = 8000$ \\
$ms(B, C) = 5000$ & $LL(B) = 2000$ \\
$ms(C, A) = 8000$ & $LL(C) = 5000$ \\
\end{longtable}
\addtocounter{table}{-1}

To audit a Minimax election, in the event that a Condorcet winner does not
exist, we can first verify which pairwise defeats represent the strongest
defeat for each candidate. In the example above, we could show that $A \succ B$
is the strongest defeat of $B$ by showing that $ms(A,B) > ms(c,B)$ for all $c
\in \mathcal{C} \setminus \{A\}$. We then need to show that $ms(A, B)$ is less
than the score of the strongest defeat of each other candidate. In the example
above, this reduces to checking that $ms(B,C) > ms(A,B)$ and $ms(C,A) >
ms(A,B)$.


\section{Smith}
\label{sec:Smith}

The \emph{Smith set} in an election refers to the smallest set of candidates
$S$ such that every candidate in $S$ defeats every candidate outside of $S$ in
a pairwise contest. For all $c \in S$ and $c' \in \mathcal{C} \setminus S$, we
have $T(c \succ c') > T(c' \succ c)$. The Smith set always exists and is well
defined\footnote{\url{https://en.wikipedia.org/wiki/Smith_set}, accessed 14 Mar
2023}. If a Condorcet winner exists, they will be the sole member of this set.
If the Smith set contains more than one candidate, IRV or Minimax can be used
to select a single winner from that set. Alternatively, all candidates in the
Smith set can be viewed as winners (if appropriate).

If we have a Condorcet winner, $w$, an audit would proceed by checking that $w$
defeats all other candidates in a pairwise contest. Otherwise, we need to first
check that the reported Smith set is correct. To do so, we first show that $T(c
\succ c') > T(c' \succ c)$ for all $c \in S, c' \not\in S$. Next, we must show
that removing any one candidate from our set would violate this condition:
$\forall c \in S, \exists c' \in S, T(c \succ c') > T(c' \succ c)$. In other
words, that each candidate in the set is defeated by another candidate in the
set. We then audit the resulting Minimax or IRV election over $S$. It may be
that a candidate in the Smith set defeats multiple candidates in the set. For
the purposes of auditing, we choose to check the defeat with the largest
margin. If there are candidates in the Smith set that tie (i.e., $T(c \succ
c')$ $ =$ $ T(c' \succ c)$ for some $c,c' \in S$), then a full manual count is
required.


\section{Kemeny-Young}
\label{sec:KY}

Under Kemeny-Young, we start by computing pairwise scores for each pair of
candidates $(c,c')$, $T(c \succ c')$. We then imagine all possible complete
orders (rankings) among the candidates in the election. We assign a score to
each ranking using the pairwise scores we have just computed. Consider the
ranking $[A, B, C]$ in an election with candidates $A$, $B$ and $C$. The
ranking score is equal to $T(A \succ B) + T(A \succ C) + T(B \succ C)$. For
each candidate $c$ that appears above another $c'$ in a ranking, we add $T(c
\succ c')$ to the ranking score. The winning candidate is the candidate ranked
first in the highest scoring ranking. 

We can view each ranking $\pi$ as an entity with a tally, $T(\pi)$. The tally
for the ranking $\pi = [A, B, C]$, for example, is $T(\pi) = T(A \succ B) + T(A
\succ C) + T(B \succ C)$. Let $\pi_r$ denote the reported highest scoring
ranking. For elections with a very small number of candidates, we can form an
audit in which we check that $T(\pi_r) > T(\pi')$ for every possible ranking
$\pi'$ with a different first-ranked candidate to $\pi_r$. For an election with
$k$ candidates, an audit with $(k-1)!$ assertions is formed. A 13 candidate
election, however, will require $4.8\times10^8$ assertions! 

While it is technically possible to form a RLA that is not a full recount for a
Kemeny-Young election, all election instances we consider in this paper have
too many candidates for this method to be practical. We are not aware of any
other methods for generating efficient RLAs for Kemeny-Young elections.


\section{Other Condorcet Methods}

Some other Condorcet methods, such as Schulze and Copeland, are not auditable
by the assertion-assorter framework as it stands. To form a RLA using this
framework, we need to be able to check that the reported winner won through a
series of comparisons over sums of ballots. Under both the Schulze and Copeland
methods, we use such comparisons to perform a meta-level reasoning step.

Under Copeland, for example, we compute each candidate $c$'s \emph{Copeland
score $CS(c)$}, which is the number of candidates $c'$ for which $T(c \succ c')
> T(c' \succ c)$ plus a half times the number of candidates $c'$ for which $T(c
\succ c') = T(c' \succ c)$. We then elect the candidate with the highest
Copeland score.

Under the Schulze method, we assign scores to each pair of candidates using the
winning votes method as per Minimax (see \autoref{sec:Minimax}). Consider a
graph where for each pair of candidates, $(c,c')$, we have a directed edge from
$c$ to $c'$ with a weight equal to $s(c, c')$. We define the strength of a path
in this graph between candidates $c$ and $c'$ as the weight of the weakest link
along this path. For each pair of candidates, we compute the strength of the
strongest path between the pair. Where there are multiple paths between the
candidates, the strongest path is the one with the largest weight on its
weakest link. If there is no path, the strength of their strongest path is set
to zero. Let us denote the strength of the strongest path between $c$ and $c'$
as $p(c, c')$. We say that $c \succ c'$ if $p(c, c') > p(c', c)$. The winner is
the candidate $w$ for which $w \succ c$ for all $c \in \mathcal{C} \setminus
\{w\}$.


\section{Results}
\label{sec:Results}

We consider the set of IRV elections conducted in the 2015 and 2019 New South
Wales (NSW) Legislative Assembly (lower house) elections, and a number of
US-based IRV elections conducted between 2007 and 2010. For each instance, we
reinterpret it as Condorcet, Ranked Pairs, Minimax and Smith elections.  We
report the estimated difficulty of conducting a comparison RLA for each
instance when viewed as an IRV election or one of the alternative Condorcet
methods.  We assume a risk limit of 5\%, and an error rate of
0.002.
RAIRE \cite{blom2019raire} is used to generate assertions for auditing each
instance as an IRV election.  Note that the intention behind the reporting of
these results is not to recommend one type of election over others, but to
demonstrate the practicality, or lack thereof, of the auditing methods we have
proposed.

We estimate the sample size (ASN) required to audit a set of assertions
$\mathcal{A}$, with a chosen SHANGRLA risk function (we used
Kaplan-Kolmogorov), through simulation. For each assertion $a \in \mathcal{A}$,
we performed 2000 simulations in which we randomly distributed errors across
the population of auditable ballots, and determined how many ballots needed to
be sampled for the risk to fall below the desired threshold. We took the median
of the resulting sample sizes to compute an anticipated sample size for the
assertion. We took the largest of these sample sizes, across $\mathcal{A}$, as
the expected sample size required for the audit. We used this process to
estimate required sample sizes for all election types.

Across all the IRV instances we considered, the same winner was declared when
the ballots were tabulated according to the rules of IRV, Condorcet, and Ranked
Pairs. All instances have a Condorcet winner. In all but one election---a US
instance, Pierce County Executive 2008---the estimated difficulty of auditing
each election as either a Ranked Pairs or Condorcet election was the same. For
Pierce County Executive 2008, checking each $T(w \succ c) - T(c \succ w) > 0$
assertion requires an estimated 627 ballot polls, the same ASN as the IRV
audit. When audited as a Ranked Pairs election, an estimated 507 ballot polls
are required. This is because in the Ranked Pairs election we are able to
declare a winner before committing to all $(w, c)$ pairs, through the use of a
transitive inference. This means that in the audit, we can avoid checking one
of the $w \succ c$ comparisons and instead check some easier assertions to
verify the transitive inference used.

\autoref{tab:ResultsNSW15-19} reports the expected sample sizes required to
audit selected elections that took place in the 2015--19 NSW Legislative
Assembly elections as either IRV, Condorcet, or Ranked Pairs.
\autoref{tab:ResultsUS} reports the same for selected US instances. Instances
for which the expected sample sizes differed substantially in the different
contexts are in bold. In general, there was no substantial difference in these
expected sample sizes when auditing an instance as an IRV, Condorcet, or Ranked
Pairs election.

\begin{table}[p]
\centering
\caption{Estimated sample sizes, expressed as both a number of ballots and
percentage of the total ballots cast, required to audit selected IRV elections
from the 2015--19 NSW Legislative Council elections (as IRV using RAIRE, Ranked
Pairs [RP], and Condorcet [CDT]). Instances where the ASN for the audits across
the different election types are substantially different are in bold. All
instances have a Condorcet winner.}
\label{tab:ResultsNSW15-19}
\begin{tabular}{l|rr|rr|rr}
\toprule
\textsc{Instance} & \multicolumn{2}{|c|}{\textsc{RAIRE IRV}} &\multicolumn{2}{|c|}{\textsc{RP}} &\multicolumn{2}{|c}{\textsc{CDT}}  \\
         & \multicolumn{2}{|c|}{\textsc{ASN (\%)}}  &  \multicolumn{2}{|c|}{\textsc{ASN (\%)}} &  \multicolumn{2}{|c}{\textsc{ASN (\%)}}  \\      
\midrule
\multicolumn{7}{l}{2015} \\
\midrule
Albury	&	31	&	0.06\%	&	28	&	0.06\%	 & 28 & 0.06\%  	\\
Auburn	&	74	&	0.15\%	&	74	&	0.15\%	 & 74 & 0.15\%  \\
Ballina	&	155	&	0.32\%	&	137	&	0.28\%	& 137 & 0.28\%  \\
Balmain	&	99	&	0.20\%	&	99	&	0.20\%	&  99& 0.20\%  	\\
Clarence	&	42	&	0.09\%	&	42	&	0.09\%	&  42& 0.09\% \\
Coffs Harbour	&	32	&	0.07\%	&	27	&	0.06\% 	&  27& 0.06\% \\
Coogee	&	137	&	0.29\%	&	137	&	0.29\%	& 137 & 0.29\%  \\
East Hills	&	1309	&	2.60\%	&	1309	&	2.60\% 	& 1309 & 2.60\%  \\
Gosford	&	3889	&	7.70\%	&	3889	&	7.70\% & 3889& 7.70\% \\
Granville	&	207	&	0.43\%	&	207	&	0.43\% & 207 & 0.43\%  \\
{\bf Lismore}	&	{\bf 1138}	&	{\bf 2.35\%}	&	{\bf 4689}	&	{\bf 9.70\%} &  {\bf 4689}	&	{\bf 9.70\%}  \\
Manly	&	15	&	0.03\%	&	15	&	0.03\% & 15 & 0.03\%  \\
Maroubra	&	35	&	0.07\%	&	35	&	0.07\% & 35 & 0.07\%  \\
Monaro	&	152	&	0.32\%	&	152	&	0.32\% & 152& 0.32\% \\
Mount Druitt	&	26	&	0.05\%	&	26	&	0.05\% & 26& 0.05\% \\
Strathfield	&	227	&	0.46\%	&	227	&	0.46\% & 227& 0.46\% \\
Summer Hill	&	45	&	0.09\%	&	45	&	0.09\% & 45 & 0.09\%  \\
Sydney	&	54	&	0.12\%	&	54	&	0.12\% & 54& 0.12\% \\
Tamworth	&	41	&	0.08\%	&	38	&	0.08\%	& 38& 0.08\%  \\
The Entrance	&	1596	&	3.18\%	&	1596	&	3.18\% & 1596& 3.18\% \\
\midrule
\multicolumn{7}{l}{2019} \\
\midrule
Albury	&	25	&	0.05\%	&	25	&	0.05\%	 & 25	&	0.05\%\\
Auburn	&	45	&	0.09\%	&	45	&	0.09\%	&45	&	0.09\% \\
Ballina	&	98	&	0.19\%	&	97	&	0.19\%	&97	&	0.19\%	\\
Balmain	&	44	&	0.09\%	&	44	&	0.09\%	&  44	&	0.09\% \\
Coogee	&	248	&	0.52\%	&	248	&	0.52\%	& 248	&	0.52\%		\\
Cronulla	&	20	&	0.04\%	&	19	&	0.04\%	& 19	&	0.04\%	 \\
Drummoyne	&	27	&	0.05\%	&	25	&	0.05\% &25	&	0.05\% 	\\
Dubbo	&	234	&	0.46\%	&	234	&	0.46\% &234	&	0.46\%	\\
East Hills	&	1173	&	2.30\%	&	1173	&	2.30\%	&1173	&	2.30\% \\
Hawkesbury & 25	&	0.05\%	&	25	&	0.05\%	& 25	&	0.05\%\\
Holsworthy	& 130	&	0.25\%	&	130	&	0.25\% &130	&	0.25\%  \\
Keira	&	19	&	0.04\%	&	19	&	0.04\% & 19	&	0.04\% \\
Kogarah	&	236	&	0.49\%	&	236	&	0.49\% & 236	&	0.49\%  \\
{\bf Lismore}	&	{\bf 1363}	&	{\bf 2.71\%}	&	{\bf 313}	&	{\bf 0.62\%} & {\bf 313}	&	{\bf 0.62\%}\\
Mulgoa	&	34	&	0.06\%	&	34	&	0.06\% & 34	&	0.06\% \\
Murray	&	129	&	0.26\%	&	129	&	0.26\% &129	&	0.26\%\\
 Newcastle	&	25	&	0.05\%	&	24	&	0.05\% & 24	&	0.05\% \\
 Oxley	&	34	&	0.07\%	&	28	&	0.06\%	 & 28	&	0.06\%	 \\
 Penrith	&	333	&	0.64\%	&	333	&	0.64\% &  333	&	0.64\% \\
\bottomrule
\end{tabular}
\end{table}

\begin{table}[!t]
\centering
\caption{Estimated sample sizes, expressed as a number of ballots and
percentage of the total ballots cast, required to audit a set of US IRV
elections (as IRV using RAIRE, Ranked Pairs [RP], and Condorcet [CDT]).
Instances where the ASNs are substantially different across election types are
in bold.  CC, CA, and CE denote City Council, County Assessor, and County
Executive.  All instances have a Condorcet winner.}
\label{tab:ResultsUS}
\begin{tabular}{l|rr|rr|rr}
\toprule
\textsc{Instance} & \multicolumn{2}{|c|}{\textsc{RAIRE IRV}} &\multicolumn{2}{|c}{\textsc{RP}} &\multicolumn{2}{|c}{\textsc{CDT}}  \\
         & \multicolumn{2}{|c|}{\textsc{ASN (\%)}}  &  \multicolumn{2}{|c}{\textsc{ASN (\%)}}  &  \multicolumn{2}{|c}{\textsc{ASN (\%)}}  \\      
\midrule
Aspen 2009 CC	&	249	&	10\%	&	249	&	10\%	  &  249& 10\%   \\
Aspen 2009 Mayor	&	100	&	3.96\%	&	142	&	5.60\%	      &  142& 5.60\% \\
Berkeley 2010 D1 CC	&	18	&	0.32\%	&	16	&	0.28\% & 16& 0.28\%	\\
Berkeley 2010 D4 CC	&	31	&	0.65\%	&	31	&	0.65\% & 31&  0.65\%  	\\
Berkeley 2010 D7 CC	&	40	&	0.96\%	&	40	&	0.96\% & 40 & 0.96\%	 \\
Berkeley 2010 D8 CC	&	17	&	0.37\%	&	17	&	0.37\% & 17 & 0.37\% \\
Oakland 2010 D4 CC	&	33	&	0.16\%	&  31	&	0.15\% & 31& 0.15\% 	\\
Oakland 2010 D6 CC	&	18	&	0.14\%	&	16	&	0.12\% &  16& 0.12\%	 \\
Oakland 2010 Mayor	&	499	&	0.42\%	&	499	&	0.42\%	           & 499 & 0.42\% \\
Pierce 2008 CC	&	70	&	0.18\%	&	70	&	0.18\%	   & 70 & 0.18\%   \\
Pierce 2008 CA	&	1153	&	0.44\%	&	1153	&	0.44\%  & 1153 & 0.44\%	 \\
Pierce 2008 CA	&	64	&	0.04\%	&	64	&	0.04\%	       & 64 &  0.04\%  \\
{\bf Pierce 2008 CE}	&	{\bf 624}	&	{\bf 0.21\%}	&	{\bf 507}	&	{\bf 0.17\%} & {\bf 624}& {\bf 0.21\%} 	\\
San Francisco Mayor 2007	&	10	&	0.01\%	&	9	&	0.01\%	 & 9 & 0.01\%  \\
\bottomrule
\end{tabular}
\end{table}

\subsection{IRV vs Ranked Pairs}

In a small number of cases---Lismore (NSW) in 2015 and 2019; and Pierce County
Executive (2008)---there was a substantial difference in the auditing
difficulty in the two contexts. For Lismore (2015), we expect to audit the IRV
election with a sample size of 1138 ballots, and the Ranked Pairs with 4689
ballots. For Lismore (2019), the situation is reversed, with an estimated 313
ballots required to audit the Ranked Pairs election and 1363 ballots for the
IRV. In the Pierce County Executive (2008) election, we expect to sample 624
ballots for the IRV election and 507 ballots for the Ranked Pairs. All three of
these instances share a common feature: when tabulated as an IRV election, the
candidate who is eliminated in the last round of elimination determines the
winner.

In Lismore (2015) [1138 ballots IRV vs 4689 ballots Ranked Pairs], the
Nationals (NAT) candidate wins. In the last round of elimination, this
candidate, alongside a Green (GRN) and a Country Labor (CLP) remain standing.
Their tallies at this stage are 20567, 12771 and 12357 votes, respectively.
Given the nature of Australian politics, we would expect the majority of
ballots sitting with the GRN at this stage to flow on to the CLP, if they were
eliminated. Conversely, if the CLP were eliminated, we would expect ballots to
flow on to both the NAT and GRN candidates. In this case, the CLP is
eliminated, and we are left with the NAT on 21660 votes and the GRN on 19310.
In the Ranked Pairs variation, the most difficult assertion to check is that
$s(NAT, CLP) > 0$, which equates to $T(NAT \succ CLP) - T(CLP \succ NAT) > 0$.
The difference in these tallies is just 186 votes.  The most difficult
assertion RAIRE has to check is that the GRN is not eliminated before the CLP
in the context where just they and the NAT remain. Here, the tallies of the GRN
and CLP differ by 414 votes. 

In Lismore (2019) [1363 ballots IRV vs 313 ballots Ranked Pairs], the CLP
candidate wins. In the last round of elimination, we have the CLP on 12860
votes, the GRN on 12500, and the NAT on 20094. In this case, the GRN is
eliminated leaving the NAT on 20712 votes, and the CLP on 21862. In the Ranked
Pairs variation, the smallest tally difference we need to check is that $T(CLP
\succ NAT) - T(NAT \succ CLP) > 0$. In the reported results, the LHS equals
1150 votes. Checking that $T(NAT \succ GRN) - T(GRN \succ NAT) > 0$ is much
simpler, with the LHS equal to 15494 votes. RAIRE has to show that the CLP
cannot be eliminated before the GRN in that last round elimination. This is
more difficult, with a difference of only 360 votes separating the two.

\subsection{Elections without a Condorcet winner}

All instances in \autoref{tab:ResultsNSW15-19} and \autoref{tab:ResultsUS} have
a Condorcet winner. We have found several ranked vote datasets that, when
treated as a single-winner election, do not have a Condorcet winner.
\autoref{tab:ResultsNCW} reports the estimated sample sizes required to audit
these instances as IRV (using RAIRE), Ranked Pairs, Minimax, and Smith. RAIRE
did not terminate in a reasonable time frame (24 hours) in 6/9 of these
instances. RAIRE relies on being able to prune large portions of the space of
alternate election outcomes through carefully chosen assertions. In these
instances, RAIRE was unable to do this.  The number of votes for each ballot
signature in the Preflib instances were multiplied by 1000 to form larger
elections.

\begin{table}[!t]
\centering
\caption{Estimated sample sizes, expressed as a number of ballots and
percentage of the total ballots cast, required to audit a set of ranked vote
election instances without a Condorcet winner (as IRV using RAIRE, Ranked
Pairs, Minimax [MM], and Smith). Instances where the ASN for the audits are
substantially different across election types are in bold. A `--' denotes that
assertions could not be found within 24 hours by the associated algorithm, and
`$\infty$' denotes that a full manual hand count is required. The overall ASN
for each Smith RLA is the maximum of the cost of verifying the Smith set and
verifying the winner from the Smith set using either Mimimax of IRV.}
\label{tab:ResultsNCW}
\begin{tabular}{l||rr||rr||rr||rr|rr|rr}
\toprule
\textsc{Instance} & \multicolumn{2}{c||}{\textsc{RAIRE IRV}} &\multicolumn{2}{c||}{\textsc{RP}}  &\multicolumn{2}{c||}{\textsc{MM}} & \multicolumn{6}{c}{\textsc{Smith}}\\
\cline{8-13}
         &  \multicolumn{2}{c||}{} & \multicolumn{2}{c||}{} & \multicolumn{2}{c||}{}  & \multicolumn{2}{c|}{\textsc{Verify SS}} & \multicolumn{2}{c|}{\textsc{Minimax}} & \multicolumn{2}{c}{\textsc{IRV}} \\
         & \multicolumn{2}{c||}{\textsc{ASN (\%)}}  &  \multicolumn{2}{c||}{\textsc{ASN (\%)}}    & \multicolumn{2}{c||}{\textsc{ASN (\%)}}  & \multicolumn{2}{c|}{\textsc{ASN (\%)}} & \multicolumn{2}{c|}{\textsc{ASN (\%)}} & \multicolumn{2}{c}{\textsc{ASN (\%)}} \\      
\midrule
\multicolumn{13}{l}{2021 NSW Local Government Elections (originally STV)}\\
\midrule
Byron &    \multicolumn{2}{c||}{--}  & \multicolumn{2}{c||}{$\infty$}  &   \multicolumn{2}{c||}{$\infty$}  & 844& 4.8\%& \multicolumn{2}{c|}{$\infty$} & 327& 1.84\% \\
Leeton &    \multicolumn{2}{c||}{--}    & 1219 & 2\% & \multicolumn{2}{c||}{$\infty$}  & \multicolumn{2}{c|}{$\infty$} & \multicolumn{2}{c|}{$\infty$} & \multicolumn{2}{c}{$\infty$}\\
Parkes &     \multicolumn{2}{c||}{--}    & \multicolumn{2}{c||}{$\infty$}   & \multicolumn{2}{c||}{$\infty$} & \multicolumn{2}{c|}{$\infty$} & \multicolumn{2}{c|}{$\infty$} & \multicolumn{2}{c}{$\infty$}\\
Yass Valley &   \multicolumn{2}{c||}{--}  & 1624 & 1.7\%   &   \multicolumn{2}{c||}{$\infty$}  & \multicolumn{2}{c|}{$\infty$} & \multicolumn{2}{c|}{$\infty$} & \multicolumn{2}{c}{$\infty$}\\
\midrule
\multicolumn{13}{l}{Preflib Election Data} \\
\midrule
ED-7-5.soi & \multicolumn{2}{c||}{--} & 2828 & 2.7\%   & 2828 & 2.7\% & \multicolumn{2}{c|}{$\infty$} & \multicolumn{2}{c|}{$\infty$} & \multicolumn{2}{c}{$\infty$}\\
ED-7-19.soi & \multicolumn{2}{c||}{$\infty$} & 563 & 0.56\%  & 563 & 0.56\% &\multicolumn{2}{c|}{$\infty$} & \multicolumn{2}{c|}{$\infty$} & \multicolumn{2}{c}{$\infty$}\\
ED-10-26.soi & 56 & 0.35\% & 113 & 0.71\% & \multicolumn{2}{c||}{$\infty$}  &\multicolumn{2}{c|}{$\infty$} & \multicolumn{2}{c|}{$\infty$} & \multicolumn{2}{c}{$\infty$} \\
ED-10-47.soi &120 & 0.71\%& \multicolumn{2}{c||}{$\infty$}  & \multicolumn{2}{c||}{$\infty$}   & 120& 0.71\%& \multicolumn{2}{c|}{$\infty$}  & 59& 0.35\%\\
ED-34-1.soi & \multicolumn{2}{c||}{--}& 4577 & 1.2\% & \multicolumn{2}{c||}{$\infty$}  & \multicolumn{2}{c|}{$\infty$} & \multicolumn{2}{c|}{$\infty$} & \multicolumn{2}{c}{$\infty$}\\
\bottomrule
\end{tabular}
\end{table}

For the Smith method, \autoref{tab:ResultsNCW} reports the difficulty of both
checking the correctness of the Smith set, and auditing the whole election
using either Minimax or IRV to find a winner from the Smith set. For the Byron
instance, with 32 candidates, the estimated sample size required to verify the
correctness of the Smith set is 844 ballots (4.8\% of the total ballots cast).
The first stage of verifying the Smith set---checking that each candidate in
the set defeats all candidates outside the set---requires 112 assertions. The
second stage---checking that each candidate within the Smith set defeats
another candidate in the set---requires 4 assertions. Verifying the winner from
the Smith set using IRV requires an estimated sample size of 327 ballots
(1.84\%). The ASN of the overall audit, with IRV used to select the winner, is
844 ballots. With Minimax used to select the overall winner, however, a full
manual count is required.

Across the ED instances, with the exception of ED-10-47, there are ties present
between candidates in the Smith set, and a full manual count is required.
Across the Leeton, Parkes, and Yass Valley instances, the second stage of
verifying the correctness of the Smith set involves at least one comparison
with a very small margin. For the ED-10-47 instance, there are tied winners
under Minimax when selecting the winner from the Smith set.

\enlargethispage{\baselineskip}


\section{Conclusion}

We have presented methods for generating assertions sufficient for conducting
RLAs for several Condorcet methods, including Ranked Pairs. We have found that
auditing a ranked vote election as IRV or Ranked Pairs requires similar
estimated sample sizes, in general. Most Ranked Pairs audits reduce to checking
that the reported winner is the Condorcet winner. Where the election does not
have a Condorcet winner, it appears that auditing the instance as a Ranked
Pairs election is generally more practical than if it were an IRV election. We
have considered auditing methods for other Condorcet methods, such as Minimax,
Smith, and Kemeny-Young, however these were generally not practical.

%
%
\bibliographystyle{splncs04}
\bibliography{BIB}

\end{document}